\documentclass[a4paper, onecolumn, accepted=2020-06-17]{quantumarticle}
\pdfoutput=1
\usepackage[utf8]{inputenc}
\usepackage{amsmath}
\usepackage{amssymb}
\usepackage{amsbsy}
\usepackage{amsthm}
\usepackage{mathtools}
\usepackage{xcolor}
\usepackage{thmtools}
\usepackage{thm-restate}
\usepackage{hyperref}
\hypersetup{ colorlinks = true, allcolors = {}}
\usepackage{cleveref}

\usepackage[numbers, sort&compress]{natbib}

\newcommand{\C}{\mathbb{C}}

\newcommand{\N}{\mathbb{N}}
\newcommand{\R}{\mathbb{R}}

  \DeclareMathOperator{\aut}{Aut}
  \newcommand{\st}{:\,} 

  \newcommand{\eps}{\epsilon}

  \DeclareMathOperator{\Tr}{\mathsf{Tr}}

  \newcommand{\id}{\mathrm{Id}}

  \newcommand{\beq}{\begin{equation}}
  \newcommand{\eeq}{\end{equation}}
  \newcommand{\beqn}{\begin{equation*}}
  \newcommand{\eeqn}{\end{equation*}}
  \newcommand{\beqr}{\begin{eqnarray}}
  \newcommand{\eeqr}{\end{eqnarray}}
  \newcommand{\beqrn}{\begin{eqnarray*}}
  \newcommand{\eeqrn}{\end{eqnarray*}}
  \newcommand{\bmline}{\begin{multline}}
  \newcommand{\emline}{\end{multline}}
  \newcommand{\bmlinen}{\begin{multline*}}
  \newcommand{\emlinen}{\end{multline*}}

\newtheorem{theorem}{Theorem}[section]
\newtheorem{lemma}[theorem]{Lemma}
\newtheorem{proposition}[theorem]{Proposition}

\newtheorem{definition}[theorem]{Definition}

\DeclareMathOperator*{\E}{\mathbb{E}}

\newcommand{\norm}[1]{\Vert #1 \Vert}

\renewcommand{\ell}{L}

\title{Quasirandom quantum channels}

\author{Tom Bannink}
\affiliation{CWI, QuSoft, Science Park 123, 1098 XG Amsterdam, Netherlands. Supported by the Gravitation-grant NETWORKS-024.002.003 from the Dutch Research Council~(NWO).}
\email{tombannink@gmail.com}
\author{Jop Bri\"et}
\affiliation{CWI, QuSoft, Science Park 123, 1098 XG Amsterdam, Netherlands. Supported by the Gravitation-grant NETWORKS-024.002.003 from the Dutch Research Council~(NWO). Additionally supported by an NWO VENI grant}
\email{j.briet@cwi.nl}
\author{Farrokh Labib}
\affiliation{CWI, QuSoft, Science Park 123, 1098 XG Amsterdam, Netherlands. Supported by the Gravitation-grant NETWORKS-024.002.003 from the Dutch Research Council~(NWO).}
\email{labib@cwi.nl}
\author{Hans Maassen}
\affiliation{QuSoft, Korteweg-de Vries Institute for Mathematics, Radboud University.}
\email{H.Maassen@math.ru.nl}
  
\begin{document}
\maketitle

\begin{abstract}
    Mixing (or quasirandom) properties of the natural transition matrix associated to a graph can be quantified by its distance to the complete graph. 
Different mixing properties correspond to different norms to measure this distance.
For dense graphs, two such properties known as spectral expansion and uniformity were shown to be equivalent in seminal 1989 work of Chung, Graham and Wilson.
Recently, Conlon and Zhao extended this equivalence to the case of sparse vertex transitive graphs using the famous Grothendieck inequality.

Here we generalize these results to the non-commutative, or `quantum', case, where a transition matrix becomes a quantum channel. 
In particular, we show that for irreducibly covariant quantum channels, expansion is equivalent to a natural analog of uniformity for graphs, generalizing the result of Conlon and Zhao.
Moreover, we show that in these results, the non-commutative and commutative (resp.) Grothendieck inequalities yield the best-possible constants.
%
\end{abstract}

\section{Introduction} \label{sec:intro}

In a seminal work~\cite{Chung1989}, Chung, Graham and Wilson --- building on work of Thomason~\cite{Thomason:1987, Thomason:1987b} --- proved that several seemingly distinct notions of quasirandomness for graphs are equivalent.
In particular, they identified seven properties found in random graphs with high probability, that always coexist simultaneously in any large dense graph. 
Two of these properties are \emph{spectral expansion} and \emph{uniformity} (defined below). 
A question of Chung and Graham~\cite{Chung:2002} on the equivalence of these two properties in \emph{sparse} graphs resulted in a line of research culminating in recent work of Conlon and Zhao~\cite{Conlon2017}, which introduced a surprising new item to the armory of combinatorics: the famous Grothendieck inequality~\cite{Grothendieck53}.
In this paper, we draw a parallel line in the context of quantum information theory, where quantum channels
take the place of graphs. 
In addition, we give a streamlined proof of the main result of~\cite{Conlon2017} and show that the use of Grothendieck's inequality yields an optimal constant.
Similarly, we show that the non-commutative Grothendieck inequality gives an optimal constant in the quantum setting.

\paragraph{Spectral expansion and uniformity.}
Spectral expansion is a linear-algebraic property given in terms of the transition matrix of a graph.
This transition matrix is the normalized adjacency matrix, which for a $d$-regular graph $G = (V,E)$ is given by $A_{uv} = e(\{u\},\{v\})/d$, where $e(S,T)$ denotes the number of edges connecting subsets $S,T\subseteq V$. 
We say that the graph~$G$ is an $(n,d,\lambda)$ graph if $|V| = n$, it is $d$-regular and all but the largest eigenvalue of~$A$, which is always~1, have modulus at most~$\lambda$.
The smallest value of~$\lambda$ for which this holds is denoted by~$\lambda(G)$.
Spectral expansion then refers to the property that~$\lambda(G)$ is much smaller than~1, in which case~$G$ is referred to as a (spectral) expander.
Expanders have many important applications in mathematics and computer science (we refer to~\cite{Hoory:2006} for an extensive survey).
One such application is in randomized algorithms, 
which can exploit the fact that a random walk on an expander rapidly mixes (i.e.\ quickly converges to its limit distribution) to significantly reduce the amount of randomness needed. 

Uniformity is a combinatorial property of the configuration of the edges.
An $n$-vertex $d$-regular graph~$G = (V,E)$ is $\epsilon$-uniform if for all $S,T\subseteq V$, 
\begin{align} \label{eq:defuniform}
\Big\vert e(S,T) - \frac{d}{n} |S| \, |T| \Big\vert \leq \epsilon d n
\end{align}
and~$\epsilon(G)$ denotes the smallest value of~$\epsilon$ for which this holds.
Uniformity then refers to the property that this parameter is much smaller than~1; trivially any graph is 1-uniform.
Intuitively, this says that for any two vertex subsets, the number of edges between those sets is close to the expected number of edges in a random graph with the same edge density.

A basic result known as the Expander Mixing Lemma~\cite{Hoory:2006} shows that for any regular graph $G$ we have $\epsilon(G) \leq \lambda(G)$,
which is to say that spectral expansion implies uniformity.
A sequence $G_n$ of $d_n$-regular graphs is called \textit{dense} if $d_n \geq \Omega(n)$, and \textit{sparse} if $d_n / n \longrightarrow 0$.
It was shown in~\cite{Chung1989} that in the dense case, a converse to the Expander Mixing Lemma $\epsilon(G_n) \leq o(1) \Rightarrow \lambda(G_n) \leq o(1)$
also holds.
In contrast, Krivelevich and Sudakov~\cite{Krivelevich2006} showed that this is false for sparse graphs, thereby answering the question posed in~\cite{Chung:2002}.
Their counterexample is not regular, however (and a later one from~\cite{Bollobas2004} is not connected).
But in~\cite{Conlon2017} it was shown that even regular sparse graphs (where $d_n \leq o(n)$) can simultaneously satisfy $\epsilon(G_n) \leq o(1)$ and $\lambda(G_n) \geq \Omega(1)$.
Surprisingly, Kohayakawa, R\"{o}dl, and Schacht~\cite{Kohayakawa2016} showed that Cayley graphs over abelian groups, including sparse ones, do again admit such a converse.
Cayley graphs are an important class of regular graphs that include for instance the famous Ramanujan graphs of Margulis~\cite{Margulis:1988} and Lubotzky, Phillips and Sarnak~\cite{Lubotzky:1988}.
Conlon and Zhao~\cite{Conlon2017} generalized this to all Cayley graphs and showed that this implies the same for all vertex-transitive graphs in general, for which they showed that $\lambda(G) \leq 4K_G \epsilon(G)$, where $1.6769\ldots \leq K_G < 1.7822\dots$ is the famous \emph{Grothendieck constant}, whose exact value is currently unknown; the bounds shown here are the best known and were shown by Davie and Reeds (independently) in~\cite{Davie:1984, Reeds:1991} and Braverman et al.\ in~\cite{Braverman:2013}, respectively.
Spectral expansion and uniformity are thus equivalent notions of quasirandomness for dense graphs and vertex-transitive graphs.


\paragraph{Quasirandomness in quantum information theory.}
A transition matrix, such as the normalized adjacency matrix of a graph, maps probability vectors\footnote{We use the convention of writing probability vectors as \emph{column} vectors intead of row vectors.} to probability vectors.
A natural non-commutative generalization of a transition matrix is a \emph{quantum channel}, a completely positive trace preserving linear map~$\Phi : M_n(\C) \to M_n(\C)$; see Section~\ref{sec:prelims} for formal definitions. Quantum channels are the most general operations on quantum systems that are physically realizable.
They encapsulate the ``classical'' transition matrices by restricting them to diagonal matrices whose diagonals form probability vectors; we discuss this in more detail in Section~\ref{sec:cz}.
In quantum information theory, general linear maps from $M_n(\C)$ to itself are referred to as \emph{superoperators}.
Since superoperators are in one-to-one correspondence with bilinear forms on~$M_n(\C)\times M_n(\C)$, they also appear in the context of (generalizations of)  Bell inequalities from physics in the form of quantum XOR games~\cite{Regev15, Cooney:2015}, as well as in combinatorial optimization~\cite{NRV:2013}.
The graph-theoretic concepts mentioned above have natural analogues for superoperators, which we discuss next.

 In independent work, Hastings~\cite{Hastings2007} and Ben-Aroya, Schwartz and Ta-Schma~\cite{Aroya2010} introduced \emph{quantum expanders} as a special class of quantum channels defined analogously to spectral expanders.
 For a superoperator~$\Phi$, the expansion parameter is given by
 \beq\label{def:super_exp}
 \lambda(\Phi) = \|\Phi - \Pi\|_{S_2 \to S_2} = \sup\big\{\|(\Phi - \Pi)(X)\|_{S_2} \st \|X\|_{S_2} \leq 1\big\},
 \eeq
 where $\Pi:X\mapsto \frac{1}{n}\mathrm{Tr}(X)\id$ is the projection onto the identity,
 $\|X\|_{S_2} = \sqrt{\langle X,X\rangle}$ is the Frobenius (or Schatten-2) norm and $\langle X,Y\rangle = \frac{1}{n}\mathrm{Tr}(Y^*X)$ is the normalized trace inner product.
 A quantum channel is an expander if~$\lambda(\Phi)$ is much smaller than~1.
Also quantum expanders found many applications, one of which is again randomness reduction, where randomness takes on the form of random unitary matrices.
Since a $k$-qubit unitary requires $4^k$ real parameters, sampling one from the uniform distribution (Haar probability measure) is very expensive. 
A 1-design is a fixed collection of unitaries $U_1,\dots, U_m$ such that the superoperator $\Phi:X\mapsto \frac{1}{m} \sum_{i=1}^m U_i X U_i^*$ exactly effects the projection~$\Pi$, thus mimicking in a finite way the Haar measure on $U(n)$.
Quantum expanders can be used to construct \textit{approximate} 1-designs, meaning that~$\Phi(X)$ and~$\Pi(X)$ are close in trace distance\footnote{The trace distance is the distance induced by the Schatten-1 norm, defined in \Cref{sec:prelims}.} instead of precisely equal. Another application is in cryptography where Ambainis and Smith~\cite{Ambainis04} used quantum expanders to construct short quantum one-time pads.
It was shown in~\cite{Hastings2007} that truly random quantum channels (given by independent Haar-uniform~$U_i$ as described above) are quantum expanders with high probability, supporting the idea that this is a notion of quasirandomness.

In this work we introduce a natural notion of uniformity for superoperators, informally given by how well they mimic the action of~$\Pi$ on projectors on subspaces, which may be thought of as generalizations of vertex subsets in graphs.
This is similar to Hasting's notion of edge expansion for quantum channels~\cite{Hastings2007}.
In particular, we say that~$\Phi$ is $\eps$-uniform if for any two subspaces~$V,W\subseteq \C^n$ with associated projections~$P_V,P_W$, it holds that 
\beq\label{def:super_uniformity}
|\langle P_V,(\Phi - \Pi)(P_W)\rangle| \leq \eps .
\eeq 
Let $\epsilon(\Phi)$ denote the smallest~$\epsilon$ for which this holds.
As we show in Section~\ref{apx:embedding}, the parameters~$\lambda(\Phi)$ and~$\eps(\Phi)$ reduce to their graphical analogs under a suitable embedding of graphs into quantum channels.

Finally, also symmetry, which in the graph-theoretic context takes the form of vertex transitivity, is an important property of quantum channels.
In particular, \emph{irreducibly covariant} quantum channels, which turn out to generalize vertex-transitive graphs (see Section~\ref{sec:cz}), play an important role in questions about the capacity of quantum channels as noisy transmitters of quantum information~\cite{Holevo2006}.
A now famous result of Hastings~\cite{Hastings2009_super} shows that the minimum output capacity in general does not have the intuitively natural property of being sub-additive under tensor products.
However, it was shown earlier by Holevo~\cite{Holevo:2002}, that the capacity is additive for the subclass of irreducibly covariant quantum channels.

\paragraph{Summary of our results.} 
In this work we make a first step in the study of the equivalence of quasirandom properties for quantum channels, or superoperators in general, and show optimality in the case of vertex-transitive graphs and covariant quantum channels.

\begin{itemize}
    \item (\Cref{sec:ncconverse}) Our main result shows that under irreducible covariance, expansion and uniformity are equivalent for superoperators.
In particular, while a simple analogue of the classical Expander Mixing Lemma implies that $\epsilon(\Phi) \leq \lambda(\Phi)$ in general, we show using a non-commutative version of Grothendieck's inequality due to Haagerup~\cite{Haagerup1985}, that for this class of superoperators, also $\lambda(\Phi) \leq 2\pi^2\epsilon(\Phi)$ always holds. This implies the same result for vertex-transitive graphs with $\C$-weighted edges, essentially proved in~\cite{Conlon2017} with the factor~2 replaced by the \emph{complex} Grothendieck constant $1.3380\ldots \leq K_G^\C \leq 1.4049\dots$.

\item (\Cref{apx:embedding}) We show that a construction of sparse regular graphs from~\cite{Conlon2017} can be embedded to give a sequence of quantum channels~$\Phi_n$ that are not irreducibly covariant and for which it holds that $\epsilon(\Phi_n) \leq o(1)$ and $\lambda(\Phi_n) \geq \Omega(1)$.

\item (\Cref{sec:dense}) We show that for \emph{randomizing} channels, a notion introduced in~\cite{Aubrun2009}, the two notions of quasirandomness are also equivalent.
This can be interpreted as a generalization of the same statement for dense graphs proved in~\cite{Chung1989}.

\item (\Cref{sec:constantc}) We show that the result of~\cite{Conlon2017} cannot be improved in the sense that the factors $4K_G$  and $\pi^2K_G^\C$ are optimal in the case of vertex-transitive graphs with $\R$-weighted and $\C$-weighted edges, respectively.

\item (\Cref{sec:constantnc}) Our work leaves open whether the factor $2\pi^2$ in our main result is optimal. However, our proof consists of two steps, the first of which gives a factor~2 and the second a factor~$\pi^2$, and we show these steps are individually optimal.
    We prove that the first step is optimal by showing that an example of Haagerup and Ito~\cite{Haagerup1995} for the non-commutative Grothendieck inequality is irreducibly covariant, which uses some representation theory of~$\mathrm{SO}(n)$.
    The optimality of the second step follows directly from a result of~\cite{Conlon2017}.
\end{itemize}


\paragraph{Acknowledgements}
We would like to thank M\={a}ris Ozols, Michael Walter and Freek Witteveen for fruitful discussions.

\section{Preliminaries}\label{sec:prelims}

Write $[n] = \{1,\dots,n\}$.
For a finite set~$S$, write $\E_{s\in S}$ for $\frac{1}{|S|} \sum_{s\in S}$. 
For a compact set~$S$, write~$C(S)$ for the set of continuous functions from $S$ to $\C$.
For a compact group $\Gamma$, write $\E_{g \in \Gamma}$ for the the integral with respect to the (unique) Haar probability measure on~$\Gamma$.

Write $M_n(\C)$ for the set of complex $n\times n$ matrices
and let $U(n) = \{X\in M_n(\C) \st X^*X = \id\}$ be the set of unitary matrices.
Here, all maps of the form $\Phi : M_n(\C) \to M_n(\C)$ are linear, and we refer to these as superoperators. 
A superoperator~$\Phi$ is \emph{unital} if $\Phi(\id)=\id$ and it is \emph{completely positive} if for all $k\in\N$ the superoperator $\id\otimes\Phi:M_k\otimes M_n\rightarrow M_k\otimes M_n$ maps positive semidefinite matrices to positive semidefinite matrices.
Completely positive superoperators that are trace preserving are called \emph{quantum channels}.

We normalize inner products so that for $x,y\in\C^n$ we define $\langle y,x\rangle = \E_{i\in[n]} \overline{y_i} x_i$ and for matrices $X,Y\in M_n(\C)$ we have $\langle Y,X\rangle = \frac{1}{n} \mathrm{Tr}[Y^* X]$.

\paragraph{Norms.}
For $p\in [1,\infty)$,  $x\in \C^n$ and $X \in M_n(\C)$, the $\ell_p$ norm and (normalized) Schatten-$p$ norm are defined by
\begin{align*}
    \norm{x}_{\ell_p} = \Big( \E_{i\in[n]} |x_i|^{p} \Big)^{1/p}
    \qquad \text{and} \qquad
    \norm{X}_{S_p} = \Big( \frac{1}{n} \mathrm{Tr}\big[ (X^* X)^{p/2} \big] \Big)^{1/p}
\end{align*}
and $\|x\|_{\ell_\infty} = \max_i|x_i|$ and $\|X\|_{S_\infty} = \sup\{|\langle Xx,y\rangle| \st \|x\|_{\ell_2},\|y\|_{\ell_2}\leq 1\}$.
Note that for the identity matrix $\id\in M_n$ we have $\norm{\id}_{S_p}=1$ for all $p\in [1,\infty]$.

\begin{proposition} \label{prop:diagonalnorm}
    Let $p \geq 1$ and let $X\in M_n(\C)$. Then $\norm{X}_{S_p} \geq \norm{(X_{11},\dots,X_{nn})}_{L_p}$.
\end{proposition}
\begin{proof}
For a vector $x\in\C^n$, denote by $\mathrm{Diag}(x)$ the $n\times n$ matrix with $x$ on the diagonal and for a matrix $X$ denote by $\mathrm{diag}(X)$ the matrix where we set the off-diagonal elements to 0.
A small computation shows that
\begin{align*}
 \E_{s\in\{\pm 1\}^n}\mathrm{Diag}(s)\, X\,\mathrm{Diag}(s)=\mathrm{diag}(X).
\end{align*}
Since the Schatten-$p$ norms are invariant under conjugation with a unitary matrix, applying the above with the triangle inequality gives 
\begin{align*}
    \norm{(X_{11},\dots,X_{nn})}_{L_p}=\norm{\mathrm{diag}(X)}_{S_p}
    \leq \E_{s\in\{\pm 1\}^n}\norm{\mathrm{Diag}(s)\, X\, \mathrm{Diag}(s)}_{S_p}=\norm{X}_{S_p}. &\qedhere
\end{align*}
\end{proof}

For $q\in [1,\infty]$, define $q'\in [1,\infty]$ to be its dual given by $\frac{1}{q} + \frac{1}{q'} = 1$.
For $p,q\in [1,\infty]$, a matrix $A\in M_n(\C)$ and a superoperator $\Phi : M_n(\C) \to M_n(\C)$, define 
\begin{align*}
\norm{A}_{\ell_p\to\ell_q} &= \sup\{|\langle y, Ax\rangle| \st \norm{x}_{\ell_p}\leq 1,\, \norm{y}_{\ell_{q'}}\leq 1\}\\
\norm{\Phi}_{S_p\to S_q} &= \sup\{|\langle Y, \Phi(X)\rangle| \st \norm{X}_{S_p}\leq 1,\, \norm{Y}_{S_{q'}}\leq 1\}.
\end{align*}
 If $G$ is a $d$-regular graph on $n$ vertices with normalized adjacency matrix~$A$, then $\lambda(G) = \|A - \frac{1}{n}J\|_{L_2\to L_2}$, where~$J$ is the all-ones matrix. Also recall from \eqref{def:super_exp} that for a superoperator $\Phi$ the expansion parameter is $\lambda(\Phi) = \|\Phi - \Pi\|_{S_2 \to S_2}$.

Also define the \emph{cut norms}  by
\begin{align*}
        \norm{A}_\mathrm{cut} &= \max \{ |\langle y, A x \rangle| \st x,y \in \{0,1\}^n \}\\
        \norm{\Phi}_\mathrm{cut} &= \sup \{ |\langle Y, \Phi(X) \rangle| \st X,Y \text{ projectors} \} .
    \end{align*}
It is then not hard to see that if $G$ is a $d$-regular graph on $n$ vertices with normalized adjacency matrix~$A$, then $\epsilon(G) = \|A - \frac{1}{n}J\|_{\mathrm{cut}}$.
Similarly, we have $\epsilon(\Phi) = \norm{\Phi -\Pi}_{\mathrm{cut}}$.

We have the following relation between these norms, the proof of which is a simple generalization of the same result from~\cite{Conlon2017} for matrices.

\begin{restatable}{lemma}{cutinftyone} \label{lemma:cutinftyone}
    For any superoperator $\Phi$, we have  $\norm{\Phi}_\mathrm{cut} \leq \norm{\Phi}_{S_\infty\to S_1} \leq \pi^2 \norm{\Phi}_\mathrm{cut}$ and $\pi^2$ is the best possible constant.
\end{restatable}
\begin{proof}
First note that the cut norm as defined above can also be written as
\begin{align}
    \label{eq:nccutnormrelaxed}
    \norm{\Phi}_\mathrm{cut} = \sup \{ |\langle Y, \Phi(X) \rangle| \st X,Y \succeq 0 \:,\; \norm{X}_{S_\infty} , \norm{Y}_{S_\infty} \leq 1 \} ,
\end{align}
    because the set $\{ X \st X \succeq 0 , \norm{X}_{S_\infty} \leq 1\}$ is the convex hull of the set of projectors. Hence, by linearity the supremum in~\eqref{eq:nccutnormrelaxed} will always be attained by projectors.

    The first inequality of the lemma follows by dropping the positive semidefinite constraint. For the second inequality, let $z$ be a complex number of norm 1, and $w$ a uniform random complex number of norm $1$. Then
    \begin{align*}
        z = \pi\; \mathbb{E}_w [ w \; 1_{\{\Re(z\bar{w})\geq 0\}} \; ].
    \end{align*}
    Note that $\E_w [ f(w) ] = \frac{1}{2\pi} \int_{0}^{2\pi} f(e^{i\theta}) d\theta$, hence the equality follows by using $\int_{-\pi/2}^{\pi/2} \cos(\theta) d\theta = 2$.
    We have $\norm{\Phi}_{S_\infty\to S_1} = \sup \{ |\langle Y, \Phi(X)\rangle| \st \norm{X}_{S_\infty} , \norm{Y}_{S_\infty} \leq 1 \}$.
    The set of matrices $X$ such that $\norm{X}_{S_\infty}\leq 1$ is the convex hull of the set of unitary matrices, so by linearity we can assume that the supremum in $\norm{\Phi}_{S_\infty\to S_1}$ is obtained by unitary $X,Y$. Unitary matrices are diagonalizable, so write $X = U A U^*$ and $Y = V B V^*$ with $U,V$ unitary and $A,B$ diagonal. Let $u,w\in\C$, $|u|=|w|=1$ be uniform random complex numbers and define diagonal matrices $A',B'$ as $A'_{ii}(w) = 1_{\{\Re(A_{ii} \bar{w})\geq 0\}}$ and $B'_{ii}(u) = 1_{\{\Re(B_{ii} \bar{u})\geq 0\}}$. By the above we have $A = \pi\; \mathbb{E}_w [ w A'(w)]$ and similar for $B$, so we have $X = \pi\; \mathbb{E}_w[ w UA'(w)U^*]$ and $Y=\pi\;\mathbb{E}_u[ u VB'(u)V^*]$.
    Now, $UA'(w)U^*$ and $VB'(u)V^*$ are projections for all values of $w$ and $u$, as required in the definition of the cut norm. Therefore
    \begin{align*}
        \norm{\Phi}_{S_\infty\to S_1}
        = |\langle Y, \Phi(X)\rangle|
        &= \pi^2 |\mathbb{E}_{u,w} \bar{u} w \langle VB'(u)V^*, \Phi(UA'(w)U^*) \rangle |\\
        &\leq \pi^2 \mathbb{E}_{u,w} |\langle VB'(u)V^*, \Phi(UA'(w)U^*) \rangle |\\
        &\leq \pi^2 \mathbb{E}_{u,w} \norm{\Phi}_\mathrm{cut} \\
        &=\pi^2\norm{\Phi}_\mathrm{cut} ,
    \end{align*}
    completing the first part of the proof.
    Conlon and Zhao show that $\pi^2$ is the best possible constant in the commutative case, using the matrix $A\in M_n(\C)$ given by $A_{st} = e^{2\pi i (s-t)/n}$. This matrix satisfies $\norm{A}_{\ell_\infty \to \ell_1} = n$ and one can show $\norm{A}_\mathrm{cut} = (\pi^{-2} + o(1))n$. By \Cref{prop:normgeneralization} in \Cref{apx:embedding}, their example can be embedded into a superoperator with the same norms so $\pi^2$ is also the best possible constant here.
\end{proof}

Define the \emph{Grothendieck norm} of of a matrix~$A\in M_n(\C)$ by
\beqn
\norm{A}_G := \sup\Big\{ \Big|\frac{1}{n}\sum_{i,j=1}^n A_{ij}\langle x_i,y_j\rangle\Big| \st d\in \N,\:\: x_i,y_j\in \C^d,\: \|x_i\|_{\ell_2} \leq 1,\, \|y_j\|_{\ell_2} \leq 1 \Big\}.
\eeqn
Then, the \emph{complex Grothendieck constant} is given by
\beqn
K_G^\C := \sup\Big\{\frac{\|A\|_G}{\|A\|_{\ell_\infty\to\ell_1}} \st n\in \N,\: A\in M_n(\C)\Big\}.
\eeqn
The current best upper and lower bounds on~$K_G^\C$ are $1.4049$~\cite{Haagerup:1987} and $1.338$~\cite{Davie:1984}, respectively.
The real version of the Grothendieck constant, denoted by $K_G$ and mentioned in the introduction, is obtained by replacing the underlying field in the above quantities by the reals.





\paragraph{Some basic group theory.}
Given a graph $G = (V,E)$, a permutation $\pi: V\to V$ is an \emph{automorphism} of $G$ if for all $u,v\in V$, we have $\{\pi(u), \pi(v)\}\in E \Leftrightarrow \{u,v\}\in E$.
The automorphisms of $G$ form a group under composition, which we call $\aut(G)$.
Then, $G$ is said to be \emph{vertex transitive} if for every $u,v\in V$, there is a $\pi\in \aut(G)$ such that $\pi(u) = v$.
For superoperators, we have the following analogous definitions.
A unitary representation of a group~$\Gamma$ on~$\C^n$ is a homomorphism from $\Gamma$ to~$U(n)$ and it is irreducible if the only subspaces of~$\C^n$ that are left invariant by the group action are the zero-dimensional subspace and $\C^n$ itself.

\begin{definition}[Irreducible covariance] \label{def:IC}
    A superoperator $\Phi:M_n(\C)\to M_n(\C)$ is \emph{irreducibly covariant} if there exist a compact group~$\Gamma$ and
    continuous irreducible unitary representations $U, V\colon \Gamma\to U(n)$ such that for all $g\in \Gamma$ and $X\in M_n(\C)$, we have 
    \begin{align*}
     \Phi(U(g) X U^*(g))= V(g)\Phi(X)V^*(g). 
    \end{align*}
\end{definition}

%

\section{Converse expander mixing lemmas} \label{sec:cz}

In this section, we prove the ``converse expander mixing lemmas'' announced in the first and third bullet in the introduction as well as the examples announced in the second bullet.
As a warm-up, we start with a proof of the commutative case due to Conlon and Zhao, which we reprove in a slightly different manner analogous to how we will prove the non-commutative case.

\subsection{Commutative case} \label{sec:czalt}
In the following, let $S$ be a compact set and~$\Gamma$ be a compact group acting continuously and transitively on $S$.
The Haar probability measure on~$\Gamma$ induces a measure on~$S$ (by pullback) according to which the $L_p$-norm (for $p\in [1,\infty)$) and inner product of~$f,g\in C(S)$ are given by
\beq\label{def:L2}
\|f\|_{L_p} = \Big(\E_{\pi\in \Gamma}\big|f\big(\pi(s_0)\big)\big|^p\Big)^{\frac{1}{p}}
\qquad
\text{and}
\qquad
\langle f,g\rangle
=
\E_{\pi\in \Gamma}\overline{f\big(\pi(s_0)\big)}g\big(\pi(s_0)\big),
\eeq
where (by transitivity)~$s_0$ can be taken to be some arbitrary but fixed element of~$S$.
We lift the action of~$\Gamma$ on $S$ to an action on $C(S)$ by precomposition, that is, for any function $f\in C(S)$ and  element~$\pi\in \Gamma$, define the function~$f^{\pi}$ by $f^{\pi}(s) := f(\pi(s))$.
Furthermore, for a linear map $A\colon C(S)\to C(S)$ define $A^\pi$ by $A^{\pi}f := (A f^\pi)^{\pi^{-1}}$ and say that~$A$ is transitive covariant with respect to~$\Gamma$ if for any $\pi\in \Gamma$ we have $A^{\pi} = A$.\footnote{In general one says $A$ is \emph{covariant} with respect to~$\Gamma$, but we say \emph{transitive} to emphasize that we require~$\Gamma$ to act transitively on~$S$.} We sometimes omit the group and simply say~$A$ is \emph{transitive covariant} if such a group~$\Gamma$ exists.

In~\cite{Conlon2017}, the following result is proved (over the real numbers) for the case~$S = [n]$, in which case transitive covariant linear maps~$A$ are simply $n\times n$ matrices which commute with the permutation matrices of a transitive subgroup $\Gamma$ of $S_n$.
However, their proof easily implies the more general version below.

\begin{restatable}[Conlon--Zhao]{theorem}{CZ_matrix}\label{thm:CZ_matrix}
    Let~$S$ be as above and let $A:C(S)\to C(S)$ be a linear map that is transitive covariant with respect to~$\Gamma$. Then,
    $$\norm{A}_{L_2\to L_2} \leq K_G^\C\norm{A}_{L_\infty\to L_1}\;.$$
\end{restatable}

Here we give a somewhat more streamlined proof of this result based on a well-known factorization version of Grothendieck's inequality~\cite{Grothendieck53} (see also~\cite{Pisier2012}), which will serve as a stepping stone to the proof of the non-commutative case.\footnote{The main difference is that in~\cite{Conlon2017}, the result is first proved for weighted Cayley graphs, after which it is shown that this implies the result for transitive covariant matrices.}
In our setting the inequality asserts the following

\begin{theorem}[Commutative Grothendieck inequality (factorization)] \label{thm:cgi}
   Let~$S$ be as above and let $A:C(S)\to C(S)$ be a linear map. Then, there exist probability measures $\lambda,\nu$ on~$S$ such that for all $f,g\in C(S)$, we have
    \begin{align*}
        \vert\langle g, A f\rangle\vert \leq K_G^\C
        \norm{A}_{\ell_\infty\to \ell_1}
        \left( \int_S |f(s)|^2 \; d\lambda(s) \right)^{1/2}
        \left( \int_S |g(s)|^2 d\nu(s) \right)^{1/2} .
    \end{align*}
\end{theorem}


\begin{proof}[Proof of Theorem~\ref{thm:CZ_matrix}]
    It follows from the triangle inequality and transitivity that
    \begin{align*}
        \vert\langle g, A f\rangle\vert &\leq
        \E_{\pi \in \Gamma} \vert\langle g, A^{\pi} f \rangle\vert  =
        \E_{\pi \in \Gamma} \vert\langle g^{\pi}, A f^{\pi} \rangle\vert.
     \end{align*}
By Theorem~\ref{thm:cgi} and the AM-GM inequality there are probability measures $\lambda,\nu$ on~$S$ such that the above right-hand side is at most
\begin{align*}
         \frac{K_G^\C \norm{A}_{\ell_\infty\to \ell_1}}{2}
        \E_{\pi \in \Gamma}
        \left( \int_S |f^\pi(s)|^2 d\lambda(s) +
        \int_S |g^{\pi}(s)|^2 d\nu(s) \right) 
        = 
        \frac{K_G^\C \norm{A}_{\ell_\infty\to \ell_1}}{2} ( \norm{f}_{\ell_2}^2 + \norm{g}_{\ell_2}^2 ) ,
    \end{align*}
    where we switched the order of the integrals (using Tonelli's theorem) and the expression~\eqref{def:L2} for the~$L_2$ norm. 
    For $\norm{f}_{\ell_2} = \norm{g}_{\ell_2} = 1$ this shows $\norm{A}_{\ell_2\to \ell_2} \leq K_G^\C \norm{A}_{\ell_\infty\to \ell_1}$.
\end{proof}

\subsection{Non-commutative case} \label{sec:ncconverse}

Our main technical result is as follows.

\begin{restatable}{theorem}{ceml} \label{thm:ceml}
    Let $\Phi:M_n(\C)\to M_n(\C)$ be an irreducibly covariant superoperator. Then,
    $\norm{\Phi}_{S_\infty\to S_1} \leq \norm{\Phi}_{S_2\to S_2} \leq 2\norm{\Phi}_{S_\infty \to S_1}$.
\end{restatable}

Since the supremum in $\norm{\Phi}_{S_\infty \to S_1}$ is taken over $X,Y$ with $S_\infty$-norm equal to 1,
the first inequality of the theorem follows from the fact that $\norm{X}_{S_2} \leq \norm{X}_{S_\infty}$.
As projectors have Schatten-$\infty$ norm~1, the first inequality also easily implies the analogue of the Expander Mixing Lemma, that is, $\epsilon(\Phi) \leq \lambda(\Phi)$,  where $\lambda(\Phi)$ and~$\epsilon(\Phi)$ are as in~\eqref{def:super_exp} and~\eqref{def:super_uniformity}, respectively; note that when $\Phi$ is irreducibly covariant, so is $\Phi - \Pi$.
The second inequality is proved at the end of this section and in \Cref{sec:constantnc} we show that the factor~2 in the theorem is optimal.
With Lemma~\ref{lemma:cutinftyone}, which relates the uniformity parameter $\epsilon(\Phi)$ to $\norm{\Phi-\Pi}_{S_\infty\to S_1}$, Theorem~\ref{thm:ceml} then immediately gives the following result stated in the introduction.


\begin{restatable}[Converse Quantum Expander Mixing Lemma]{corollary}{corceml} \label{cor:ceml}
    Let $\Phi : M_n(\C) \to M_n(\C)$ be an irreducibly covariant superoperator. Then, $\lambda(\Phi) \leq 2\pi^2\epsilon(\Phi)$.
\end{restatable}

In this non-commutative setting we use the following analog of Theorem~\ref{thm:cgi} (a factorization version of the non-commutative Grothendieck inequality), proved by Haagerup in~\cite{Haagerup1985}; see also~\cite{Pisier2012}.
A density matrix is a positive semidefinite matrix with trace equal to 1.

\begin{theorem}[Haagerup] \label{thm:grothendieckfactorization}
    Let $\Phi\colon M_n(\C)\to M_n(\C)$ be a superoperator. Then, there exist density matrices $\rho_1,\rho_2,\sigma_1,\sigma_2$ such that for any $X,Y\in M_n(\C)$, we have
    \begin{align} \label{eq:grothendieckfactorization}
        \vert\langle Y,\Phi(X)\rangle\vert \leq
        \norm{\Phi}_{S_\infty\to S_1}
        \left( \mathrm{Tr}[\rho_1 X^* X] + \mathrm{Tr}[\rho_2 X X^*] \right)^{1/2}
        \left( \mathrm{Tr}[\sigma_1 Y^* Y] + \mathrm{Tr}[\sigma_2 Y Y^*] \right)^{1/2} .
    \end{align}
\end{theorem}

We also use the following lemma.

\begin{restatable}{lemma}{ictransitive} \label{lemma:ictransitive}
    Let $\Gamma$ be a compact group.
    Then, a unitary representation $U:\Gamma \to U(n)$ is irreducible if and only if for any $X\in M_n(\C)$, we have
    $$\E_{g \in \Gamma} U(g) X U(g)^* = \mathrm{Tr}(X)\frac{1}{n} \id.$$
\end{restatable}
\begin{proof}
    By Schur's lemma, if $U$ is an irreducible representation, then for $T \in M_n(\C)$
    \begin{align*}
        \Big[ \forall g \in \Gamma \quad U(g) T U(g)^* = T \Big] \iff \Big[ \exists \lambda\in \C \quad T = \lambda \: \id \Big] .
    \end{align*}
    Let $T_X = \E_{g\in \Gamma} U(g) X U(g)^*$, then by the group structure we have $U(g) T_X U(g)^* = T_X$ for all $g\in \Gamma$. Therefore, if $U$ is irreducible then $T_X = \lambda_X \: \id$. By taking the trace, it follows that $\lambda_X = \mathrm{Tr}(X)/n$. In the other direction, if $U$ is reducible then there exists a projector $P$ onto an irreducible subspace that is left invariant, i.e. $U(g) P U(g)^* = P$ for all $g\in \Gamma$, so $T_P \neq \lambda \id$.
\end{proof}

\begin{proof}[Proof of Theorem~\ref{thm:ceml}]
    Denote by $\Gamma$ and $U, V\colon \Gamma\to U(n)$ the group and irreducible representations such that $\Phi$ is irreducibly covariant with respect to $\Gamma$ (see Definition~\ref{def:IC}).
    For any $X,Y \in M_n(\C)$ write $X_g = U(g) X U^*(g)$ and $Y_g = V(g) Y V^*(g)$, then we have
    \begin{align*}
        \vert \langle Y , \Phi(X) \rangle \vert
         &=\E_{g\in \Gamma}\vert \langle Y_g, \Phi(X_g) \rangle \vert .
    \end{align*}
    By Theorem~\ref{thm:grothendieckfactorization} and the AM-GM inequality, there exist density matrices $\rho_1,\rho_2,\sigma_1,\sigma_2$ such that the right hand side is bounded from above by
    \begin{align*}
        \frac{1}{2} \norm{\Phi}_{S_\infty\to S_1} \E_{g\in \Gamma} \left(
        \mathrm{Tr}[\rho_1 X_g^* X_g] +
        \mathrm{Tr}[\rho_2 X_g X_g^*] +
        \mathrm{Tr}[\sigma_1 Y_g^* Y_g] +
        \mathrm{Tr}[\sigma_2 Y_g Y_g^*] \right) .
    \end{align*}
    By Lemma~\ref{lemma:ictransitive} we have $\E_{g\in \Gamma} X_g^* X_g = \E_{g\in \Gamma} U(g) X^* X U^*(g) = \frac{1}{n} \mathrm{Tr}[X^* X] \mathrm{Id} = \norm{X}_{S_2}^2 \mathrm{Id}$. Let $\rho$ be a density matrix, then $\E_{g\in \Gamma} \mathrm{Tr}[ \rho X_g^* X_g] = \norm{X}_{S_2}^2$.
    The same holds for $\E_{g\in \Gamma} \mathrm{Tr}[\rho X_g X_g^*]$ but with $U$, and for $Y$ with $V$, so we see that the above quantity is equal to
    \begin{align*}
        \norm{\Phi}_{S_\infty\to S_1} \left(\Vert X\Vert_{S_2}^2+\Vert Y\Vert_{S_2}^2\right).
    \end{align*}
    If $\norm{X}_{S_2} = \norm{Y}_{S_2} = 1$ we obtain $\norm{\Phi}_{S_2\to S_2} \leq 2 \norm{\Phi}_{S_\infty\to S_1}$.
\end{proof}

%
%
%

\subsection{Embedding graphs into quantum channels} \label{apx:embedding}

In this subsection, we elucidate the claim that quantum channels generalize graphs and prove the result stated in the second bullet in the introduction, namely that there are non-irreducible quantum channels for which a converse expander  mixing lemma does not hold.

We consider the following embeddings.
For $A\in M_n(\C)$, define $\Phi_A : M_n(\C) \to M_n(\C)$ as
\begin{align}\label{def:embedding}
    \Phi_A(X) = \sum_{i,j} A_{ij} X_{jj} E_{ii} ,
\end{align}
where $E_{ij}$ is the matrix with a single~$1$ at position~$(i,j)$. 
When $A$ is a transition matrix, i.e., its column sums are 1, then it is not hard to see that~$\Phi_A$ is completely positive and trace preserving and that~\mbox{$\Phi_{\frac{1}{n}J} = \Pi$}.
Several other ways exist to create quantum expanders from expander graphs, see for example~\cite{Hastings2009} and~\cite{Harrow2008}, but as we show below, our embedding given above carries over all relevant properties of the graph we consider here.

Conlon and Zhao~\cite{Conlon2017} give an infinite sequence of $d$-regular graphs~$G_n$ that are $o(1)$-uniform but for which $\lambda(G_n) \geq 1/2$. 
Combined with the following proposition, this immediately gives the result stated in the second bullet in the introduction.

\begin{restatable}{proposition}{normgeneralization} \label{prop:normgeneralization}
    Let $A\in M_n(\C)$ and $p,q\in [1,\infty]$. Then, for~$\Phi_A$ as in~\eqref{def:embedding}, we have 
    \beqn
    \norm{\Phi_A - \Pi}_{S_p\to S_q} = \norm{A - \frac{1}{n}J}_{\ell_p\to \ell_q}\quad \text{and}\quad  \norm{\Phi_A - \Pi}_{\mathrm{cut}} = \norm{A - \frac{1}{n}J}_\mathrm{cut}.
    \eeqn
\end{restatable}

\begin{proof}
    Let $B = A - \frac{1}{n}J$, then $\Phi_A - \Pi = \Phi_B$.
    By compactness and definition of $\norm{\cdot}_{S_p\to S_q}$ we can assume there is an $X \in M_n(\C)$ such that $\norm{\Phi_B}_{S_p\to S_q} = \norm{\Phi_B(X)}_{S_q} / \norm{X}_{S_p}$.
    Write $X = \mathrm{diag}(x) + X_\mathrm{other}$ where $x\in\C^n$ is the diagonal of $X$, and $X_\mathrm{other}$ are the off-diagonal entries. Note that by definition of $\Phi_B$ we have $\Phi_B(X)=\Phi_B(\mathrm{diag}(x)) = \mathrm{diag}(Bx)$.
    By definition of Schatten norms, $\norm{\mathrm{diag}(x)}_{S_p} = \norm{x}_{\ell_p}$ and by Proposition~\ref{prop:diagonalnorm} we have $\norm{X}_{S_p} \geq \norm{x}_{\ell_p}$.
    We have
    \begin{align*}
        \norm{B}_{\ell_p\to \ell_q} \geq \frac{\norm{Bx}_{\ell_q} }{ \norm{x}_{\ell_p} } \geq \frac{ \norm{\mathrm{diag}(Bx)}_{S_q} }{ \norm{X}_{S_p} } = \frac{ \norm{\Phi_B(X)}_{S_q} }{ \norm{X}_{S_p} } = \norm{\Phi_B}_{S_p\to S_q}
    \end{align*}
    Now let $y\in\C^n$ be such that $\norm{B}_{\ell_p\to \ell_q} = \norm{By}_{\ell_q} / \norm{y}_{\ell_p}$. Then 
    \begin{align*}
        \norm{\Phi_B}_{S_p\to S_q} \geq \frac{\norm{\Phi_B(\mathrm{diag}(y))}_{S_q} }{ \norm{\mathrm{diag}(y)}_{S_p} } = \frac{ \norm{\mathrm{diag}(By)}_{S_q} }{ \norm{y}_{\ell_p} } = \frac{ \norm{By}_{\ell_q} }{ \norm{y}_{\ell_p} } = \norm{B}_{\ell_p\to \ell_q} .
    \end{align*}
    This proves the first part.

    The cut norm of a matrix takes the supremum over $x,y\in \{0,1\}^n$. Instead we can relax this to $x,y \in [0,1]^n$, since by linearity the supremum will always be attained by the extreme points. Similarly, for the superoperator case, we use Equation~\eqref{eq:nccutnormrelaxed}. Then, there exist $x,y\in [0,1]^n$ such that $\norm{B}_\mathrm{cut} = |\langle Bx, y\rangle|$. We have $\mathrm{diag}(x),\mathrm{diag}(y)\succeq 0$ and $\norm{\mathrm{diag}(x)}_{S_\infty},\norm{\mathrm{diag}(y)}_{S_\infty} \leq 1$. Therefore
    \begin{align*}
        \norm{\Phi_B}_\mathrm{cut} \geq |\langle \mathrm{diag}(y), \Phi_B(\mathrm{diag}(x)) \rangle| = |\langle \mathrm{diag}(y), \mathrm{diag}(Bx)\rangle| = |\langle y, Bx\rangle| = \norm{B}_\mathrm{cut} .
    \end{align*}
    In the other direction, let $X,Y\in M_n(\C)$ such that $X,Y\succeq 0$ and $\norm{X}_{S_\infty} , \norm{Y}_{S_\infty} \leq 1$. Define $x,y$ to be the diagonals of $X,Y$, i.e. $x_i = X_{ii}$ and $y_i = Y_{ii}$. By Proposition~\ref{prop:diagonalnorm} we have $\norm{x}_{\ell_\infty},\norm{y}_{\ell_\infty}\leq 1$. Since $X,Y\succeq 0$ we know all diagonal entries of $X$ and $Y$ are real and non-negative, so we have $x,y\in [0,1]^n$. We conclude
    \begin{align*}
        \norm{B}_\mathrm{cut} \geq |\langle y, Bx \rangle| = |\langle \mathrm{diag}(y), \mathrm{diag}(Bx) \rangle| = |\langle Y, \Phi_B(X) \rangle| = \norm{\Phi_B}_\mathrm{cut},
    \end{align*}
    completing the proof.
\end{proof}
Note that $\norm{A - \frac{1}{n}J}_{\ell_2\to \ell_2}$ is the second largest eigenvalue in absolute value of the matrix $A$, so spectral expansion is preserved under the embedding of graphs into quantum channels. Also, uniformity is preserved since the cut-norm does not change.

The following proposition shows that the embedding~\eqref{def:embedding} preserves transitivity.
This shows that our Theorem~\ref{thm:ceml} generalizes the main result of~\cite{Conlon2017}, albeit with a slightly worse constant.

\begin{restatable}{proposition}{transitivegeneralization} \label{prop:transitivegeneralization}
    For any $A\in M_n(\C)$, $A$ is vertex transitive if and only if $\Phi_A$ is irreducibly covariant.
\end{restatable}

\begin{proof}
    Suppose~$A$ is vertex transitive.
    Let $\pi \in\mathrm{Aut}(A)$ be a permutation and $P_\pi \in M_n(\C)$ be the associated permutation matrix, so that $P_\pi A P_\pi^* = A$.
    Then,
    \begin{align*}
        \Phi_A (P_\pi X P_\pi^*)
        &= \sum_{i,j} A_{ij} (P_\pi X P_\pi^*)_{jj} E_{ii} \\
        &= \sum_{i,j} A_{ij} X_{\pi^{-1}(j)\pi^{-1}(j)} E_{ii} \\
        &= \sum_{i,j} A_{i\pi(j)} X_{jj} E_{ii}\\
        &= \sum_{i,j} A_{\pi(i)\pi(j)} X_{jj} E_{\pi(i)\pi(i)} \\
        &= \sum_{i,j} A_{\pi(i)\pi(j)} X_{jj} (P_\pi E_{ii} P_\pi^*) = P_\pi \Phi_A(X) P_\pi^*.
    \end{align*}
    This shows that for all $\pi \in \mathrm{Aut}(A)$ we have $\Phi_A(P_\pi X P_\pi^*)=P_\pi \Phi_A(X) P_\pi^*$.

Let~$\mathbb{T} = \{c\in \C\st |c| = 1\}$ be the complex unit circle.
    For $\alpha\in\mathbb{T}^n$, define $U_\alpha:= \mathrm{diag}(\alpha)$. 
    We have $U_\alpha E_{ii} U_\alpha^* = |\alpha_i|^2 E_{ii} = E_{ii}$ and $(U_\alpha X U_\alpha^*)_{ii} = |\alpha_i|^2 X_{ii} = X_{ii}$. Therefore
    \begin{align*}
        \Phi_A (U_\alpha X U_\alpha^*)
        = \sum_{i,j} A_{ij} (U_\alpha X U_\alpha^*)_{jj} E_{ii}
        = \sum_{i,j} A_{ij} X_{jj} U_\alpha E_{ii} U_\alpha^* =  U_\alpha \Phi_A(X) U_\alpha^* .
    \end{align*}
    We combine these two observations as follows. First we have that
    \begin{align*}
        \left( \E_{\alpha \in \mathbb{T}^n} U_\alpha X U_\alpha^* \right)_{ij}
        = \E_{\alpha \in \mathbb{T}^n} \alpha_i X_{ij} \overline{\alpha_j}
        = \int_{0}^{2\pi} \int_{0}^{2\pi} e^{i\theta_i} X_{ij} e^{-i\theta_j} \; d\theta_i d\theta_j = X_{ii} \delta_{ij}
    \end{align*}
    If $A$ is vertex transitive then for all $x\in\C^n$ we have $\E_{\pi \in \mathrm{Aut}(A)} P_\pi \, \mathrm{diag}(x)\, P_\pi^* = (\E_{i} x_i)\; \id$.
    Therefore
    \begin{align*}
        \E_{\substack{\pi \in \mathrm{Aut}(A)\\ \alpha \in \mathbb{T}^n}} (P_\pi U_\alpha) X (P_\pi U_\alpha)^*
        = \E_{\pi \in \mathrm{Aut}(A)} P_\pi \left(\E_{\alpha \in \mathbb{T}^n} U_\alpha X U_\alpha^* \right) P_\pi^*
        = \frac{\mathrm{Tr}(X)}{n} \id.
    \end{align*}
    Letting $G\subset M_n(\C)$ be the subgroup generated by the $U_\alpha$ and $P_\pi$ for $\pi \in \mathrm{Aut}(A)$, we see that for any $g\in G$ $$\Phi_A(gXg^*)=g\Phi_A(X)g^*$$ and by the previous equation and Lemma~\ref{lemma:ictransitive}, $G$ acts irreducibly on $\C^n$ (and it is unitary).
    This proves $\Phi$ is irreducibly covariant with respect to the group $G$ with equal representations.

    For the other direction, let $U: G \to U(n)$ be the irreducible representation such that $\Phi_A$ is irreducibly covariant, i.e. $\Phi_A(U(g) X U^*(g)) = U(g) \Phi_A (X) U^*(g)$ for all $g\in G$. Define $P_g \in M_n(\C)$ as $(P_g)_{ij} = |U(g)_{ij}|^2$ so that $( U(g) E_{jj} U(g)^* )_{ii} = (P_g)_{ij}$.
    Then
    \begin{align*}
        A_{kl} = \mathrm{Tr}[ E_{kk} \Phi_A(E_{ll}) ]
        &= \mathrm{Tr}[ U(g) E_{kk} U(g)^*  \; \Phi_A( U(g) E_{ll} U(g)^* ) ] \\
        &= \sum_{ij} A_{ij} (P_g)_{jl} (P_g)_{ik} = (P_g^T A P_g)_{kl} ,
    \end{align*}
    showing $P_g^T A P_g = A$. Since $U(g)$ is unitary, $P_g$ is doubly stochastic so by Birkhoff's Theorem~$P_g$ is a convex combination of permutation matrices, i.e., $P_g = \E_{i} \Pi_i$ for some (not necessarily uniform) probability distribution and where $\Pi_i$ is a permutation matrix. We have
    \begin{align*}
        A_{kl} = (P_g^T A P_g)_{kl} = \E_{i} \E_{j} (\Pi_i^T A \Pi_j)_{kl} = \E_{i} \E_{j} A_{\pi_i(k)\:\pi_j(l)} .
    \end{align*}
    Since $A$ is $\{0,1\}$-valued, it follows that if $A_{kl}=1$ then all elements of the convex combination on the right-hand side must be~$1$, and if $A_{kl}=0$ then all elements of the right hand side must be~0. Therefore, for all~$i$ we have $\Pi_i^T A \Pi_i = A$.
    By irreducibility, we have for all $k,l$ that
    \begin{align*}
        \frac{1}{n} = \frac{\mathrm{Tr}[ E_{kk} ]}{n} \id_{ll}
        &= \left( \E_{g\in G} U(g) E_{kk} U^*(g) \right)_{ll}
        = \E_{g\in G} \left\vert U(g)_{lk} \right\vert^2 ,
    \end{align*}
    showing $\E_{g\in G} (P_g)_{lk} = 1/n$. It follows that there is a $g\in G$ such that $(P_g)_{lk}> 0$. Decomposing $P_g$ into permutation matrices shows there is a $\Pi \in \mathrm{Aut}(A)$ such that $\Pi_{lk} = 1$. This holds for all $k,l$, proving the lemma.
\end{proof}

\subsection{Randomizing superoperators} \label{sec:dense}

We prove the following analogue of one of the results from~\cite{Chung1989} showing that for any \mbox{$d$-regular} graph~$G$, it holds that $\lambda(G) \leq \left(2\epsilon(G)/\delta^2\right)^{1/4}$, where $\delta = d/n$ is the edge density. 
This in particular establishes a tight relation between spectral expansion and uniformity for sequences of graphs with \mbox{$\delta_n \geq \Omega(1)$}.
For $A\in M_n(\C)$, we have $\norm{A}_{\ell_1\to \ell_\infty} = n \sup_{ij} |A_{ij}|$, and for an $n$-vertex $d$-regular graph with normalized adjacency matrix~$A$ we have $\sup_{ij} |A_{ij}| = \frac{1}{d}$ so $\norm{A-J/n}_{\ell_1\to \ell_\infty} = \frac{1}{\delta}-1$ with $J$ being the all-ones matrix. Therefore, a sequence of graphs with normalized adjacency matrices $A_n$ is dense exactly when $\norm{A_n-J_n/n}_{\ell_1\to \ell_\infty} \leq \mathcal{O}(1)$, where $J_n$ is the all-ones $n$ by $n$ matrix .

Let $\Pi$ be the projector onto the identity matrix. A superoperator~$\Phi$ is said to be $\eta$-\emph{randomizing} if $\norm{\Phi-\Pi}_{S_1\to S_\infty} \leq \eta$, which when $\eta \leq \mathcal O(1)$, may thus be seen as an analogue of density.
Note that by \Cref{prop:normgeneralization} the embedding of any dense graph is $\mathcal{O}(1)$-randomizing.

\begin{restatable}{proposition}{ncdense} \label{prop:ncdense}
    Let $\Phi : M_n(\C) \to M_n(\C)$ be a superoperator that is $\mathcal O(1)$-randomizing.
    Then, $\lambda(\Phi) \leq \mathcal{O} (\epsilon(\Phi)^{1/4})$.
\end{restatable}

To prove Proposition~\ref{prop:ncdense}, we require the following lemma.
\begin{restatable}{lemma}{ncgowers} \label{lemma:ncgowers}
    Let $\Phi : M_n(\C) \to M_n(\C)$ be a superoperator and let $C = \norm{\Phi}_{S_1\to S_\infty}$.
    Then we have $\norm{\Phi}_{S_2 \to S_2} \leq \Big( C^3 \norm{\Phi}_{S_\infty \to S_1} \Big)^{1/4}$.
\end{restatable}
\begin{proof}
    Note that by definition of $C$ we have $| \langle Q, \Phi(P)\rangle | \leq C \norm{Q}_{S_1} \norm{P}_{S_1}$.
    Let $X,Y \in M_n(\C)$ be such that $\langle Y, \Phi(X)\rangle = \norm{\Phi}_{S_2 \to S_2}$ with $\norm{X}_{S_2} = \norm{Y}_{S_2} = 1$. Write $X = \frac{1}{n}\sum_{i=1}^n \lambda_i P_i$ and $Y = \frac{1}{n}\sum_{i=1}^n \mu_i Q_i$ with $P_i,Q_i$ rank-1 matrices with $\norm{Q_i}_{S_1} = \norm{P_i}_{S_1} = 1$. We have $\norm{\lambda}_{\ell_2} = \norm{\mu}_{\ell_2} = 1$ and by applying Cauchy-Schwarz twice,
\begin{align*}
    |\langle Y, \Phi(X) \rangle|^4
    &= \Big| \E_{ij} \lambda_i \mu_j \langle Q_j , \Phi (P_i) \rangle \Big|^4 \\
    &\leq \Big(\E_i \lambda_i^2 \Big)^{2} \; \Big( \E_{i} \big|\E_{j}\mu_j \langle Q_j , \Phi (P_i) \rangle \big|^2 \Big)^2 \\
    &= \Big( \E_{i,j,j'}\mu_j \mu_{j'} \langle  Q_j , \Phi(P_i) \rangle \langle P_i , \Phi^*(Q_{j'}) \rangle \Big)^2 \\
    &\leq \Big( \E_{j,j'} \mu_j^2 \mu_{j'}^2 \Big) \Big( \E_{j,j'} \Big| \E_i \langle Q_j, \Phi(P_i) \rangle \langle P_i, \Phi^* (Q_{j'}) \rangle \Big|^2 \Big) \\
    &= \E_{i,i',j,j'} \langle Q_j, \Phi (P_i) \rangle \langle P_i, \Phi^*(Q_{j'}) \rangle \langle Q_{j'}, \Phi(P_{i'})\rangle \langle P_{i'}, \Phi^*(Q_j) \rangle ,
\end{align*}
where all indices are averaged from $1$ to $n$. Now we see
\begin{align*}
    |\langle Y, \Phi(X) \rangle|^4
    &\leq \E_{i,j} \langle Q_j, \Phi (P_i)\rangle \Big\langle \E_{j'} \langle Q_{j'}, \Phi (P_i)\rangle Q_{j'}, \Phi \big( \E_{i'} \langle P_{i'}, \Phi^* (Q_j)\rangle P_{i'} \big) \Big\rangle \\
    &\leq \E_{i,j} |\langle Q_j, \Phi (P_i)\rangle| \; \norm{\Phi}_{S_\infty\to S_1} \; \norm{\E_{j'} \langle Q_{j'}, \Phi (P_i)\rangle  Q_{j'} }_{S_\infty} \; \norm{ \E_{i'} \langle P_{i'}, \Phi^* (Q_j)\rangle  P_{i'} }_{S_\infty} \\
    &\leq \E_{i,j} |\langle Q_j, \Phi (P_i)\rangle| \; \norm{\Phi}_{S_\infty\to S_1} \; \max_{j'} |\langle Q_{j'}, \Phi (P_i)\rangle | \; \max_{i'} | \langle Q_j, \Phi (P_{i'})\rangle| \\
    &\leq C^3 \norm{\Phi}_{S_\infty\to S_1}. & \qedhere
\end{align*}
\end{proof}

\begin{proof}[Proof of \Cref{prop:ncdense}]
    Let $\Pi(X) = \frac{1}{n} \mathrm{Tr}[X] \id$ be the projector on to the identity. By assumption, we have $\norm{\Phi-\Pi}_\mathrm{cut} = \epsilon(\Phi)$. Define $C = \norm{\Phi-\Pi}_{S_1\to S_\infty}$. Using Lemma~\ref{lemma:cutinftyone} and Lemma~\ref{lemma:ncgowers} applied to $\Phi-\Pi$ we find $\norm{\Phi-\Pi}_{S_2\to S_2} \leq ( C^3 \pi^2 \epsilon(\Phi) )^{1/4}$.
\end{proof}

\section{Optimality of constants} \label{sec:constants}

\subsection{Commutative case} \label{sec:constantc}

In this section we prove the fourth bullet point in our introduction.
Theorem~\ref{thm:CZ_matrix} shows that~$K_G^\C$ bounds the ratio of the $L_2\to L_2$ and $L_\infty\to L_1$ norms, and Lemma~\ref{lemma:cutinftyone} (the matrix version) shows that~$\pi^2$ bounds the ratio of the $L_\infty\to L_1$ norm and the cut norm. We now prove the optimality of the combined inequality.

Let $S^{m-1} = \{x\in \C^m \st \|x\|_{L_2} = 1\}$ denote the $(m-1)$-dimensional unit sphere endowed with its Haar probability measure~$\mu$.
\begin{theorem} \label{thm:optimalcombined}
    For any $\epsilon>0$ there exist positive integers $m,k$ and a transitive covariant linear map $M:C(S^{m-1}\times[k])\to C(S^{m-1}\times [k])$ such that $\norm{M}_{L_2\to L_2} \geq (\pi^2 K_G^\C - \epsilon) \norm{M}_{\mathrm{cut}}$.
\end{theorem}
The optimality of $\pi^2$ between the $L_\infty\to L_1$ norm and the cut norm is already covered in Lemma~\ref{lemma:cutinftyone}. We show that $K_G^\C$ is optimal in the sense that Theorem~\ref{thm:CZ_matrix} cannot be improved (despite the fact that the exact value of the Grothendieck constant~$K_G^\C$ is unknown). We do this in Lemma~\ref{lemma:optimalKG} below. Then in \Cref{lemma:pilift} we show that any map can be lifted to one on a bigger space with appropriately bounded cut norm. The combination of these lemmas proves our theorem.

In the introduction we also mentioned the optimal constant $4 K_G$ in the case where the field is $\R$ instead of $\C$. The proofs below still apply in this case, with only small modifications.

\begin{lemma} \label{lemma:optimalKG}
    For any $\epsilon > 0$ there exists a positive integer $m$ and a transitive covariant linear map $B:C(S^{m-1}) \to C(S^{m-1})$ such that $\|B\|_{L_2\to L_2} \geq (K_G^\C - \epsilon)\|B\|_{L_\infty\to L_1}$.
\end{lemma}

\begin{proof}
    By definition of the Grothendieck constant, for any $\epsilon > 0$ there exists an $n\in \N$ and a linear map $A \in M_n(\C)$ such that $\norm{A}_{\mathrm{G}} \geq (K_G^\C - \epsilon) \norm{A}_{L_\infty \to L_1}$.
    This map $A$ might not be transitive covariant, so from it we will now construct a transitive covariant linear map $B:C(S^{2n-1}) \to C(S^{2n-1})$ such that  $\norm{B}_{\ell_\infty\to \ell_1} \leq \norm{A}_{\ell_\infty \to \ell_1}$ and $\norm{B}_{\ell_2\to \ell_2} \geq \norm{A}_\mathrm{G}$.
This idea is based on a lemma found in~\cite{Briet2011}.


    Let $x^i , y^j \in S^{2n-1}$ be the vectors that attain the Grothendieck norm for $A$, which can always be assumed to be $2n$-dimensional since there are only~$2n$ of them, so
    \begin{align*}
        \norm{A}_G = \Big|\frac{1}{n} \sum_{i,j} A_{ij} \langle x^{i} , y^{j} \rangle\Big|.
    \end{align*}
    Define the map $B$ by
    \begin{align*}
        \langle f, B(g)\rangle = \frac{1}{n} \sum_{i,j} A_{ij} \int_{U(2n)} f( U x^{i} ) g(U y^{j}) dU.
    \end{align*}
    To bound $\norm{B}_{\ell_\infty\to \ell_1}$ we have to bound $|\langle f, B(g)\rangle|$ for $f,g : S^{2n-1} \to [-1,1]$.
    By the triangle inequality,
    \begin{align*}
        |\langle f, B(g)\rangle|
        &\leq \int_{U(2n)} \Big| \frac{1}{n} \sum_{i,j} A_{ij}  f( U x^{i} ) g(U y^{j}) \Big| dU 
        \leq \int_{U(2n)} \norm{A}_{\ell_\infty\to \ell_1}
          dU \leq \norm{A}_{\ell_\infty\to \ell_1} .
    \end{align*}
    Now for each $i\in[2n]$ let $f_i\in C(S^{2n-1})$ be given by $f_i(x) = x_i$ (i.e. the $i$-th coordinate). Then,
    \begin{align*}
        \frac{1}{2n} \sum_{i=1}^{2n}\langle f_i,B(f_i)\rangle
        &\leq  \frac{1}{2n} \sum_{i=1}^{2n}\|B\|_{L_2\to L_2} \|f_i\|_{L_2}^2\\
        &= \|B\|_{L_2\to L_2}\int_{S^{2n-1}} \frac{1}{2n}\sum_{i=1}^{2n}x_i^2 d\mu(x)\\
        &= \|B\|_{L_2\to L_2}.
    \end{align*}
    On the other hand,
    \begin{align*}
        \frac{1}{2n} \sum_{i=1}^{2n}\langle f_i,B(f_i)\rangle
        = \frac{1}{n} \sum_{i,j} A_{ij} \int_{U(2n)} \langle U x^{i} , U y^{j} \rangle dU
        = \frac{1}{n} \sum_{i,j} A_{ij} \langle x^{i} , y^{j} \rangle = \norm{A}_G ,
    \end{align*}
    so we conclude $\norm{B}_{L_2\to L_2} \geq \norm{A}_G$.
    We will show $B$ is transitive covariant with respect to $\Gamma = U(2n)$. To show $B$ is invariant, we have to prove that for all $V\in U(2n)$ we have $\langle f^V , B(g^V)\rangle = \langle f, B(g)\rangle$. Indeed,
    \begin{align*}
        \langle f^V , B(g^V)\rangle
        &= \frac{1}{n} \sum_{i,j} A_{ij} \int_{U(2n)} f( V U x^{i} ) g( V U y^{j}) dU \\
        &= \frac{1}{n} \sum_{i,j} A_{ij} \int_{U(2n)} f( U' x^{i} ) g( U' y^{j}) dU' = \langle f, B(g)\rangle ,
    \end{align*}
    which completes the proof.
\end{proof}

\begin{lemma} \label{lemma:pilift}
    Let $S$ be any compact set and let $B : C(S) \to C(S)$ be a linear map. For any $\epsilon>0$ there exists a $k\in \N$ and a linear map $M : C(S\times [k])\to C(S\times[k])$ such that
    \begin{align*}
        \frac{\norm{M}_\mathrm{cut}}{\norm{M}_{L_2\to L_2}} \leq \big( \frac{1}{\pi^2} + \epsilon \big) \frac{\norm{B}_{L_\infty \to L_1}}{\norm{B}_{L_2 \to L_2}}
    \end{align*}
    and if $B$ is transitive covariant then so is $M$.
\end{lemma}
\begin{proof}
    We will choose $k$ large enough, to be determined later.
    For any $f,g \in C(S \times [k])$ define $f^i \in C(S)$ as $f^i(s) := f(s,i)$, and similar for $g^i$.    
    Define $\omega = e^{2\pi i / k}$.
    Define a linear map $M : C(S\times [k]) \to C(S\times [k])$ as
    \begin{align*}
        \big( M(f) \big)(t,j) := \frac{1}{k} \sum_{i=1}^k \omega^{i-j} B(f^i)(t) , \quad \text{for }t\in S\text{ and } j \in [k] .
    \end{align*}
    We then have
    \begin{align*} 
        \langle g, M (f)\rangle_{S\times [k]}
        &= \frac{1}{k^2} \Big\langle \sum_i \omega^i g^i, B \big( \sum_j \omega^j f^j \big) \Big\rangle_S
    \end{align*}
    where one factor of $\frac{1}{k}$ comes from our normalization of the inner product. This implies
    \begin{align} \label{eq:splitnorm}
        \big|\langle g, M (f)\rangle_{S\times[k]}\big|
        &\leq \norm{B}_{L_\infty \to L_1}
        \Big\Vert \frac{1}{k} \sum_{i=1}^k \omega^i g^i \Big\Vert_{L_\infty}
        \Big\Vert \frac{1}{k} \sum_{j=1}^k \omega^j f^j \Big\Vert_{L_\infty} .
    \end{align}
    If $f,g\in C(S\times[k])$ are the $[0,1]$-valued functions that attain the cut norm of~$M$, then by \eqref{eq:splitnorm}
    \begin{align*}
        \norm{M}_\mathrm{cut} &\leq \Big(\frac{1}{\pi^2} + \epsilon \Big) \norm{B}_{L_\infty \to L_1} ,
    \end{align*}
    where we used \Cref{lemma:pibound} to bound $\Big\Vert \frac{1}{k} \sum_{i=1}^k \omega^i g^i \Big\Vert_{L_\infty}$.

    Let $u,v \in C(S)$ with $\norm{u}_{L_2} = \norm{v}_{L_2} = 1$ be such that $\norm{B}_{L_2\to L_2} = \langle v , B(u)\rangle_S$.
    Now define $f_{(u)},g_{(v)} \in C(S\times[k])$ as $f_{(u)}(s,i) := \omega^{-i} u(s)$ and $g_{(v)}(s,i) :=  \omega^{-i} v(s)$, which also have $L_2$-norm equal to 1. We then see
    \begin{align*}
        \norm{M}_{L_2\to L_2} \geq
        \big\langle g_{(v)}, M (f_{(u)}) \big\rangle_{S\times[k]} = \langle v, B (u) \rangle_S = \norm{B}_{L_2\to L_2} .
    \end{align*}
    The combination of these observations completes the first part of the proof.
    Now assume~$B$ is transitive covariant with respect to~$\Gamma$, so $B(f^\pi)(\pi^{-1}(s)) = B(f)(s)$ for all $s\in S$ and $\pi \in \Gamma$. Define a new group $\Gamma'$ as the cartesian product $\Gamma' = \Gamma \times \mathbb{Z}_k$. For $(\pi,m)\in \Gamma'$ define the action $(\pi,m) : S\times [k] \to S\times [k]$ as $(\pi,m)(s,i) = (\pi(s), i+m)$. By entering $f^{(\pi,m)}$ into the definition of~$M$ it follows that $M^{(\pi,m)} = M$, so~$M$ is transitive covariant with respect to~$\Gamma'$, completing the proof.
\end{proof}

\begin{lemma} \label{lemma:pibound}
    Let $\epsilon > 0$, then there exists a $k_0 \in \N$ such that
    for all $k\geq k_0$ and $x \in [0,1]^k$ we have
    \begin{align*}
        \Big\vert \frac{1}{k} \sum_{j=1}^k e^{2\pi i \, j/k} x_j \Big\vert \leq \frac{1}{\pi} + \epsilon .
    \end{align*}
\end{lemma}
\begin{proof}
    First let $k_0$ be arbitrary, to be determined later and $k\geq k_0$.
    Define $y \in [-1,1]^k$ as $y_i = 2 x_i  - 1$, then
    \begin{align*}
        \Big\vert \frac{1}{k} \sum_{j=1}^k e^{2\pi i \, j / k} x_j \Big\vert
        = \frac{1}{2} \Big\vert \frac{1}{k} \sum_{j=1}^k e^{2\pi i\, j/k} y_j \Big\vert
        = \frac{1}{2} e^{2\pi i \phi} \frac{1}{k} \sum_{j=1}^k e^{2\pi i\,j/k} y_j  .
    \end{align*}
    In the first equality we used that $\sum_{j=1}^k e^{2\pi i \, j / k} = 0$.
    In the second equality we used that there exists a $\phi$ such that the full expression becomes real and positive. Since $e^{i \theta} = \cos(\theta) + i \sin(\theta)$ and the full expression is real, we know the $\sin$ component vanishes and therefore
    \begin{align*}
        \frac{1}{2} \frac{1}{k} \sum_{j=1}^k e^{2\pi i (\phi + j/k)} y_j
        &= \frac{1}{2} \frac{1}{k} \sum_{j=1}^k \cos(2\pi (\phi + j/k)) y_j .
    \end{align*}
    Now note that $\cos(2\pi (\phi + j/k)) y_j \leq \big|\cos(2\pi (\phi + j/k))\big|$ and hence
    \begin{align*}
        \frac{1}{2} \frac{1}{k} \sum_{j=1}^k \big| \cos(2\pi (\phi + j/k)) \big|
        \;\; \overset{k\to \infty}{\longrightarrow} \;\;
        \frac{1}{2} \int_{0}^{1} \big| \cos\big(2\pi (\phi + x)\big) \big| dx = \frac{1}{\pi} .
    \end{align*}
    This completes the proof.
\end{proof}

\subsection{Non-commutative case} \label{sec:constantnc}

In the non-commutative case we show optimality of \Cref{thm:ceml}.
By \Cref{lemma:cutinftyone}, the factor~$\pi^2$ between the cut-norm and $S_\infty\to S_1$-norm is also optimal.
In contrast with the commutative case, our work leaves the optimality of the combined inequality in \Cref{cor:ceml} as an open problem. 
Straightforward analogues of the techniques employed in \Cref{lemma:pilift} did not follow through in the non-commutative case.

\begin{restatable}{proposition}{opt_ceml} \label{prop:opt_ceml}
For any $\eps > 0$, there exists a positive integer~$n$ and an irreducibly covariant superoperator~$\Phi:M_n(\C)\to M_n(\C)$ such that $\norm{\Phi}_{S_2\to S_2} \geq (2-\eps)\norm{\Phi}_{S_\infty \to S_1}$.
\end{restatable}

One of the forms of the non-commutative Grothendieck inequality, equivalent to Theorem~\ref{thm:grothendieckfactorization}, is the following~\cite{Pisier2012}. Let $\Phi : M_n(\C) \to M_n(\C)$ be a linear map and $x_i, y_j \in M_n(\C)$ finite sets of matrices. Then,
\begin{align} \label{eq:ncgrothendieck}
        \Big| \sum_i \langle x_i, \Phi(y_i)\rangle \Big| \leq K_G' \norm{\Phi}_{S_\infty\to S_1}
        \left( \frac{ \norm{\sum_i x_i^* x_i} + \norm{\sum_i x_i x_i^*} }{2} \!\cdot\!
               \frac{ \norm{\sum_i y_i^* y_i} + \norm{\sum_i y_i y_i^*} }{2} \right)^{1/2}
\end{align}
where $K_G' \leq 2$ and the norms on the right hand side are operator norms $\norm{\cdot}_{S_\infty}$.
To show tightness, i.e. $K_G' \geq 2$, Haagerup and Itoh~\cite{Haagerup1995} (see~\cite{Pisier2012} for a survey) gave an explicit family of operators for which~\eqref{eq:ncgrothendieck} gives a lower bound of $K_G'$ approaching $2$. 
We will show that slight modifications of these operators are irreducibly covariant, which proves Proposition~\ref{prop:opt_ceml}. 
It is instructive to repeat their construction.
The proof uses techniques familiar in the context of the antisymmetric Fock space, but our proof is self contained.

\begin{lemma}[\cite{Haagerup1995}] \label{lemma:pisierconstruction}
    For each $n\in\mathbb{N}$ there exists a $d\in \N$ and a linear map $\Phi : M_{d}(\C) \to M_{d}(\C)$ with sets of matrices $\{x_i\}$, $\{y_i\}$ such that~\eqref{eq:ncgrothendieck} yields $K_G' \geq (2n+1)/(n+1)$.
\end{lemma}
\begin{proof}
    Let $H~=~\C^{2n+1}$ and consider the antisymmetric $k$-fold tensor product $H^{\wedge k}$ which is a linear subspace of the $k$-fold tensor product $H^{\otimes k}$.
    A basis of $H^{\wedge k}$ is formed by vectors $e_{i_1} \wedge e_{i_2} \wedge \cdots \wedge e_{i_k}$ with $i_1 < \cdots < i_k$ where the $e_i$ are standard basis vectors of $H$. Here $\wedge$ is the wedge product or exterior product, which has the property $x\wedge y = - y \wedge x$ and is given by $x \wedge y = x\otimes y - y \otimes x$, for $x,y\in H$.
    We will consider $k=n$ and $k=n+1$ so that the dimension of $H^{\wedge k}$ is $d = \binom{2n+1}{n}$ for both $k=n$ and $k=n+1$.

    For $1\leq i \leq (2n+1)$, define $c_i : H^{\wedge n} \to H^{\wedge (n+1)}$ as $c_i (x) := e_i \wedge x$, which physicists call the fermionic creation operator. Its adjoint $c_i^* : H^{\wedge (n+1)} \to H^{\wedge n}$ is known as the annihilation operator. By the antisymmetric property, $c_i(x)=0$ whenever $e_i$ was present in $x$, i.e., when $x = e_i \wedge x'$. The operator $c_i c_i^*$, also known as the number operator, is a projector onto the space spanned by basis vectors in which $e_i$ is present. The operator $c_i^* c_i$ is a projector onto the space where $e_i$ is \emph{not} present. Since there are always $(n+1)$ vectors present in $H^{\wedge (n+1)}$ and $(n+1)$ vectors \emph{not} present in $H^{\wedge n}$, we have
    \begin{align*}
        \sum_{i=1}^{2n+1} c_i c_i^* = (n+1) \mathrm{Id}_{H^{\wedge (n+1)}}
        \quad\text{and}\quad
        \sum_{i=1}^{2n+1} c_i^* c_i = (n+1) \mathrm{Id}_{H^{\wedge n}} .
    \end{align*}
    We will now argue that
    \begin{align}
        \langle c_i, c_j \rangle := \frac{1}{d}\Tr( c_i^* c_j ) &= \delta_{i,j} \frac{n+1}{2n+1} , \label{eq:traceci} \\
        \norm{ \sum_{i=1}^{2n+1} \alpha_i c_i }_{S_1} &= \norm{\alpha}_{\ell_2} \frac{n+1}{\sqrt{2n+1}} \qquad \text{for } \alpha \in \C^{2n+1} . \label{eq:1normci}
    \end{align}
    The $\delta_{i,j}$ in~\eqref{eq:traceci} follows because $\langle x , c_i^* c_j x \rangle = 0$  for any $x = e_{k_1} \wedge \dots \wedge e_{k_n}$ when $i\neq j$.
    The factor $\frac{n+1}{2n+1}$ follows by taking the trace of one of the sums above and noting that by symmetry in $i$, every term of the sum must have the same trace.
    To prove~\eqref{eq:1normci}, first note that for any unitary $U \in \mathrm{U}(2n+1)$ we have
    \begin{align}
        U^{\otimes(n+1)} \cdot c_i \cdot (U^{\otimes n})^{-1} = \sum_{j} U_{ji} c_j \label{eq:ciconjugation} ,
    \end{align}
    which can be shown by proving it for all basis states:
    \begin{align*}
        U^{\otimes (n+1)} c_i (U^{\otimes  n})^{-1} (e_{k_1}\wedge ... \wedge e_{k_n})
        &= U^{\otimes(n+1)} c_i (U^{-1} e_{k_1}\wedge ... \wedge U^{-1} e_{k_n}) \\
        &= U^{\otimes(n+1)} (e_i \wedge U^{-1} e_{k_1}\wedge ... \wedge U^{-1} e_{k_n}) \\
        &= (U e_i \wedge e_{k_1}\wedge ... \wedge e_{k_n}) \\
        &= (\sum_{j} U_{ji} e_j \wedge e_{k_1}\wedge ... \wedge e_{k_n}) \\
        &= \sum_{j} U_{ji} c_j (e_{k_1}\wedge ... \wedge e_{k_n}) .
    \end{align*}
    The trace-norm is unitarily invariant, so~\eqref{eq:ciconjugation} implies $\norm{c_i}_{S_1} = \norm{\sum_j U_{ji} c_j}_{S_1}$. Since $c_i^* c_i$ is a projector, we have $\sqrt{c_i^* c_i} = c_i^* c_i$ and hence $\norm{c_i}_{S_1} = \frac{1}{d}\Tr(c_i^* c_i)$. Now let $\alpha \in \C^{2n+1}$ with $\sum_i |\alpha_i|^2 = 1$, then there is a unitary $U\in \mathrm{U}(2n+1)$ such that the $i$-th row of $U$ is $\alpha$. Note that $\norm{\alpha}_{\ell_2} = 1/\sqrt{2n+1}$ since we use normalized $\ell_2$-norms, which implies~\eqref{eq:1normci}.

    Since the dimensions of $H^{\wedge n}$ and $H^{\wedge (n+1)}$ are equal, we can identify the space of linear maps $L(H^{\wedge n}, H^{\wedge (n+1)})$ with $M_d(\C)$ (by choosing bases for $H^{\wedge n}$ and $H^{\wedge (n+1)}$), and define the following operator $\Phi : M_d(\C) \to M_d(\C)$,
    \begin{align*}
        \Phi(x) = \sum_{i=1}^{2n+1} \langle c_i , x \rangle \, c_i .
    \end{align*}
    Consider~\eqref{eq:ncgrothendieck} for $\Phi$ with $x_i = y_i = c_i$.
    For the left hand side, note that by~\eqref{eq:traceci} we have
    \begin{align*}
        \Big|\sum_{j=1}^{2n+1} \langle c_j , \Phi(c_j) \rangle \Big|
        = \Big| \sum_{i,j=1}^{2n+1} \langle c_i, c_j \rangle \, \langle c_j, c_i \rangle \Big|
        = \frac{(n+1)^2}{2n+1} .
    \end{align*}
    For the right-hand side of~\eqref{eq:ncgrothendieck}, we require $\norm{\Phi}_{S_\infty \to S_1} = \sup_{\norm{x}_{S_\infty} = 1} \norm{\Phi(x)}_{S_1}$.
    For any $x\in M_d(\C)$, define $v^{(x)} \in \C^{2n+1}$ as $v^{(x)}_i = \langle c_i, x\rangle$.
    Note that $\norm{v}_{\ell_2} = \sup_{\norm{\alpha}_{\ell_2}=1} |\langle v, \alpha\rangle|$.
    First apply~\eqref{eq:1normci} to obtain
    \begin{align*}
        \norm{\Phi(x)}_{S_1}
        = \norm{\sum_{i=1}^{2n+1} \langle c_i,  x\rangle c_i}_{S_1}
        = \norm{v^{(x)} }_{\ell_2} \frac{n+1}{\sqrt{2n+1}}
        = \sup_{\norm{\alpha}_{\ell_2}=1} |\langle v^{(x)}, \alpha \rangle| \frac{n+1}{\sqrt{2n+1}} .
    \end{align*}
    Using~\eqref{eq:1normci} again, we compute $\sup_{\norm{x}_{S_\infty}=1} |\langle v^{(x)}, \alpha \rangle|$ for arbitrary $\alpha$ with $\norm{\alpha}_{\ell_2}=1$,
    \begin{align*}
        \sup_{\norm{x}_{S_\infty} = 1} |\langle v^{(x)}, \alpha \rangle|
        = \sup_{\norm{x}_{S_\infty} = 1} \frac{1}{2n+1} \big| \langle x, \sum_i \alpha_i c_i \rangle \big|
        = \frac{1}{2n+1}\norm{ \sum_i \alpha_i c_i }_{S_1} = \frac{n+1}{(2n+1)\sqrt{2n+1}} .
    \end{align*}
    We obtain $\norm{\Phi}_{S_\infty\to S_1} = (n+1)^2/(2n+1)^2$. Now~\eqref{eq:ncgrothendieck} yields $\frac{(n+1)^2}{2n+1} \leq K_G' \frac{(n+1)^2}{(2n+1)^2} \cdot (n+1)$ and therefore $\frac{2n+1}{n+1} \leq K_G'$.
\end{proof}

We use the following fact from \cite[Theorem~19.14]{fulton-harris}, about the representations of the odd dimensional complex special orthogonal groups on wedge products of \emph{complex} vector spaces.

\begin{lemma}\label{lemma:antisymirreducible}
    Let $n, k \in \mathbb{N}$, $N:=2n+1$ and let $R_k : \mathrm{SO}(N,\C) \to \mathrm{GL}((\C^{N})^{\wedge k})$ be given by $A \mapsto A^{\otimes k}$. This representation is irreducible.
\end{lemma}
Below, we actually need that the \emph{real} special orthogonal group $\mathrm{SO}(N,\R)$ acts irreducibly on the same anti-symmetric space. 
Fortunately, this is implied by Lemma~\ref{lemma:antisymirreducible}; see~\cite[pp.~439]{fulton-harris}.
We will also use the fact that $R_k$ and $R_{N-k}$ are \emph{unitarily} equivalent to each other. 
This is the content of the following proposition~\cite[Proposition~IX.10.4]{simon}.

\begin{proposition}\label{prop:equivalence-of-irreps}
For positive integer~$n$ and $N = 2n+1$ and $k\in \{1,\dots,N\}$,
let~$R_k$ be the representation as in lemma~\ref{lemma:antisymirreducible}. 
Then, there exists an isometry $V_k\colon (\C^{N})^{\wedge k}\to (\C^{N})^{\wedge (N-k)}$ such that 
\begin{align*}
V_k R_k(A)=R_{N-k}(A) V_k,\quad \forall A\in \mathrm{SO}(N,\R).
\end{align*}
\end{proposition}

\begin{proof}[Proof of Proposition~\ref{prop:opt_ceml}]
Let $d$ be the dimension of $(\C^{N})^{\wedge n}$ and let~$\Phi:M_d(\C)\to M_d(\C)$ be as in the proof of Lemma~\ref{lemma:pisierconstruction}.
    For each~$k\in \N$, let $R_k : \mathrm{SO}(N,\R) \to \mathrm{GL}(H^{\wedge k})$ be the representation $A\mapsto A^{\otimes k}$, which is irreducible by \Cref{lemma:antisymirreducible}. Define, for notational convenience, $\pi:=R_{n+1}$ and $\rho:=R_n$.
    We first show that for all $A\in \mathrm{SO}(N,\R)$, we have
    \begin{align}\label{relation-Phi}
        \Phi( \pi(A) x \rho^*(A) ) = \pi(A) \; \Phi(x) \; \rho^*(A) .
    \end{align}
    For the left-hand side, note that
    \begin{align*}
    	\Phi( \pi(A) x \rho^*(A) ) &=
		\sum_i\big\langle c_i, \pi(A) x\rho^*(A)\big\rangle\, c_i\\
        &= \sum_i\big\langle \pi(A)^*c_i\rho(A), x\rangle\, c_i\\ 
        &= \sum_i\Big\langle \sum_j A_{ij}c_j,x\Big\rangle\, c_i \\ 
        &= \sum_{ij} A_{ij}\langle c_j,x\rangle\, c_i, 
    \end{align*}
    where we used~\eqref{eq:ciconjugation} from the proof of Lemma~\ref{lemma:pisierconstruction} and noting that $\mathrm{SO}(N,\R)\subset \mathrm{U}(N)$ is a subgroup.
    Using~\eqref{eq:ciconjugation} again for the right-hand side, we have
    \begin{align*}
        \pi(A) \; \Phi(x) \; \rho^*(A)
        &= 
        \sum_i \langle c_i, x\rangle\, \pi(A)c_i\rho^*(A)\\
        &=
        \sum_{i} \langle c_i,x\rangle\, \sum_j A_{ji} c_j\\
        &=
        \sum_{ij} A_{ij} \langle c_j, x\rangle\, c_i.
    \end{align*}
    which proves~\eqref{relation-Phi}.
    
     Define a new superoperator $\Phi'\colon M_d(\C)\to M_d(\C)$ by
\begin{align*}
\Phi'(x)= \Phi(xV^*)V,
\end{align*}
where $V:=V_{n+1}$ is the isometry as in Proposition~\ref{prop:equivalence-of-irreps} (we view~$V$ as a matrix in $M_d(\C)$ by choosing basis).
We first note that this~$\Phi'$ might also be used in Lemma~\ref{lemma:pisierconstruction} to show that the non-commutative Grothendieck constant is 2, since Schatten-norms are unitarily invariant. Hence, if we show that $\Phi'$ is irreducibly covariant, we are done. This follows from the following computation, where we use~\eqref{relation-Phi} and the fact that $V\pi(A)=\rho(A)V$ for all $A\in \mathrm{SO}(N,\R)$:
    \begin{align*}
    \Phi'\big(\pi(A)x\pi(A)^*\big)
    &=
    \Phi\big(\pi(A)x\pi(A)^*V^*\big)V\\
    &=
    \Phi\big(\pi(A)xV^*\rho(A)^*\big)V\\
    &\stackrel{\eqref{relation-Phi}}{=}
    \pi(A)\; \Phi(xV^*)\; \rho(A)^*V\\
    &= \pi(A)\; \Phi(xV^*)\; V\pi(A)^*\\
    &= \pi(A)\; \Phi'(x)\; \pi^*(A),
    \end{align*}
    where the second-last line follows since $\rho(A)^* = V\pi(A)^*V^*$.
    Hence, $\Phi'$ is irreducibly covariant with respect to the irreducible representation $\pi$ of $\mathrm{SO}(N,\R)$.
\end{proof}

\bibliographystyle{plainnat}
\bibliography{bibliography}

\end{document}